\documentclass[preprint,12pt]{elsarticle}
\usepackage{comment}

\usepackage{amsthm, amsmath, amssymb, bm}
\newtheorem{proposition}{Proposition}
\newtheorem{lemma}{Lemma}

\usepackage{booktabs}
\usepackage{multirow}

\usepackage{algorithm,algpseudocode}

\journal{Applied Mathematics and Computation}

\begin{document}

\begin{frontmatter}

\title{A Parallel PageRank Algorithm For Undirected Graph}

\author[auth1]{Qi Zhang}
\author[auth1]{Rongxia Tang}
\author[auth1]{Zhengan Yao}
\author[auth2]{Zan-Bo Zhang \corref{cor1}}
\address[auth1]{Department of Mathematics, Sun Yat-sen University, China}         
\address[auth2]{School of Statistics and Mathematics, Guangdong University of Finance and Economics, China}
\cortext[cor1]{Zan-Bo Zhang, zanbozhang@gdufe.edu.cn}

\begin{abstract}
As a measure of vertex importance according to the graph structure, PageRank has been widely applied in various fields. While many PageRank algorithms have been proposed in the past decades, few of them take into account whether the graph under investigation is directed or not. Thus, some important properties of undirected graph\textemdash symmetry on edges, for example\textemdash is ignored. In this paper, we propose a parallel PageRank algorithm specifically designed for undirected graphs that can fully leverage their symmetry. Formally, our algorithm extends the Chebyshev Polynomial approximation from the field of real function to the field of matrix function. Essentially, it reflects the symmetry on edges of undirected graph and the density of diagonalizable matrix. Theoretical analysis indicates that our algorithm has a higher convergence rate and requires less computation than the Power method, with the convergence rate being up to 50\% higher with a damping factor of $c=0.85$. Experiments on six datasets illustrate that our algorithm with 38 parallelism can be up to 39 times faster than the Power method.
\end{abstract}


\end{frontmatter}

\section{Introduction}
PageRank\cite{brin2012reprint} is used to measure the importance of vertices according to graph structure. Generally, a vertex has higher importance if it is linked by more vertices or the vertices linking to it have higher importance themselves. Despite originating from search engine, PageRank has been applied in many fields such as social network analysis, chemistry, molecular biology and so on\cite{gleich2015pagerank,brown2017pagerank,Yin2020change}. 

In the past decades, plenty of PageRank algorithms had been proposed, among which Power method\cite{haveliwala2003computing} and Monte Carlo(MC) method\cite{avrachenkov2007monte} were two main approaches. Existing algorithms\cite{kamvar2004adaptive,kamvar2003extrapolation,wu2007power,zhu2005distributed} have advantages respectively, however, most of them do not distinguish the graph under investigation is directed or not. For example, Power method treats any graph as directed, while undirected graph widely exists in applications such as power grids, traffic nets and so on. Compared with directed graph, the most significant property of undirected graph is the symmetry on edges. Hence, simply treating undirected graph as directed omitted this property. A few algorithms based on the MC method benefited from symmetry. For example, Sarma\cite{sarma2013fast} accelerated MC method on undirected graph by borrowing ideas from distributed random walk\cite{das2013distributed}, Luo\cite{2019Distributed,luo2020improved} proposed Radar Push on undirected graph and decreased the variance of result. Thus, taking advantage of symmetry is a feasible solution of designing PageRank algorithm for undirected graph.

PageRank vector can be obtained by computing $(\bm{I}-c\bm{P})^{-1}\bm{p}$, where $\bm{I}$ is identity matrix, $\bm{P}$ is the possibility transition matrix of the graph and $\bm{p}$ is the personalized vector. The state-of-the-art PageRank algorithm, i.e., Forward Push(FP)\cite{2006Local}, approximates $(\bm{I}-c\bm{P})^{-1}\bm{p}$ by $\sum \limits_{i=0}^{k}(c\bm{P})^{i}\bm{p}$. From the perspective of polynomial approximation, $\sum \limits_{i=0}^{k}(cx)^{i}$ might not be the optimal approximating polynomial of $(1-cx)^{-1}$ as polynomial with higher degree is required to obtain an approximation with less error. Thus, constructing the optimal polynomial is a natural idea. 

Chebyshev Polynomial\cite{2011Chebyshev} is a set of orthogonal polynomials with good properties, by which the optimal uniform approximating polynomial of any real, continuous and square integrable function can be constructed. Directly extending Chebyshev Polynomial approximation from the field of real function to the field of matrix function is always infeasible. However, if $\bm{p}$ can be linearly represented by $\bm{P}$'s eigenvectors, then $(\bm{I}-c\bm{P})^{-1}\bm{p}$ can be converted to the linear combination of $\bm{P}$'s eigenvectors, where the coefficients are functions of the corresponding eigenvalues. By applying the Chebyshev Polynomial approximation in the coefficients, a new approach of computing PageRank can be obtained. To this end, two conditions required that (1)$\bm{P}$ has enough linearly independent eigenvectors and (2)$\bm{P}$'s eigenvalues are all real. The first condition is sufficed by the density of diagonalisable matrix. For undirected graph, the second condition can be ensured by the symmetry. Motivated by this, a parallel PageRank algorithm specially for undirected graph is proposed in this paper. The contributions are as below. 
\begin{enumerate}
	\item[(1)] We present a solution of computing PageRank via the Chebyshev Polynomial approximation. Based on the recursion relation of Chebyshev Polynomial, PageRank vector can be calculated iteratively.  
	\item[(2)] We proposes a parallel PageRank algorithm, i.e., the Chebyshev Polynomial Approximation Algorithm(CPAA). Theoretically, CPAA has higher convergence rate and less computation amount compared with the Power method.
	\item[(3)] Experimental results on six data sets suggest that CPAA markedly outperforms the Power method, when damping factor $c=0.85$, CPAA with 38 parallelism requires only 60\% iteration rounds to get converged and can be at most 39 times faster. 
\end{enumerate}

The remaining of this paper are as follows. In section 2, PageRank and Chebyshev Polynomial are introduced briefly. In section 3, the method of computing PageRank via Chebyshev Polynomial approximation is elaborated. Section 4 proposes CPAA and the theoretical analysis as well. Experiments on six data sets are conducted in section 5. A summarization of this paper is presented in Section 6.

\section{Preliminary}
In this section, PageRank and Chebyshev Polynomial are introduced briefly. 

\subsection{PageRank}
Given graph $G=(V,E)$, where $V=\{v_{1},v_{2},\cdots,v_{n}\}$, $E=\{(v_{i},v_{j}): i,j=1,2,\cdots,n\}$ and $\vert E \vert=m$. Let $\bm{A}=(a_{ij})_{n \times n}$ denote the adjacency matrix where $a_{ij}$ is 1 if $(v_{j},v_{i}) \in E$ and 0 else. Let $\bm{P}=(p_{ij})_{n \times n}$ denote the probability transition matrix where 
\begin{center}
	\small
	$
	p_{ij} = \left\{ 
	\begin{array}{ll}
		a_{ij} / \sum \limits_{i=1}^{n} {a_{ij}}, & if \quad \sum \limits_{i=1}^{n} {a_{ij} \neq 0},\\
		0,                                        & else.
	\end{array}
	\right.
	$
\end{center}
Let $\bm{d}=(d_{1},d_{2},...,d_{n})^{T}$ where $d_{i}$ is 1 if $v_{i}$ is dangling vertex and 0 else. Let $\bm{p}=(p_{1},p_{2},\cdots,p_{n})^{T}$ denote the $n$-dimensional probability distribution. Let $\bm{P}^{'}=\bm{P}+\bm{p}\bm{d}^{T}$ and $\bm{P}^{''}=c\bm{P}^{'}+(1-c)\bm{p}\bm{e}^{T}$, where $c\in(0,1)$ is damping factor and $\bm{e}=(1,1,...,1)^{T}$. Denote by $\bm{\pi}$ the PageRank vector. According to\cite{Berkhin2005A}, let $\bm{p}=\frac{\bm{e}}{n}$, one of PageRank's definitions is
\begin{equation}\label{eq21}
  \small
  \bm{\pi}=\bm{P}^{''}\bm{\pi},\pi_{i}>0,\sum \limits_{i=1}^{n}\pi_{i}=1.
\end{equation}
Based on random walk model, a random walk starts at random vertex, with probability $c$ walks according to graph and with probability $1-c$ terminates, then $\bm{\pi}$ is the probability distribution of random walk terminating at each vertex.

If $G(V,E)$ is undirected graph, then $\bm{d}=\bm{0}$ and $\bm{P}=\bm{P}^{'}$. Let $\bm{D}=diag\{d_{1},d_{2},\cdots,d_{n}\}$ where $d_{i}=\sum \limits_{j=1}^{n}{a_{ij}}$, then $d_{i}>0$ and $\bm{P}=\bm{A}\bm{D}^{-1}$. From Formula \eqref{eq21}, we have  
\begin{center}
	\small
	$(\bm{I}-c\bm{P})\bm{\pi}=(1-c)\bm{p}$.
\end{center}
Since $\vert\vert \bm{P} \vert\vert_{2}=1$, $\bm{I}-c\bm{P}$ is invertible, thus   
\begin{center}
	\small
	$\bm{\pi}=(1-c)(\bm{I}-c\bm{P})^{-1}\bm{p}$.
\end{center}
In fact, $(1-c)(\bm{I}-c\bm{P})^{-1}\bm{p}$ is the algebraic form of Forward Push. With initial mass distribution $\bm{p}$, each vertex reserves $1-c$ proportion of mass receiving from its source vertices and evenly pushes the remaining $c$ proportion to its target vertices, then $\bm{\pi}$ is the final mass distribution with  
\begin{equation}\label{eq22}
	\small
	\pi_{i}=\frac{\bm{e}^{T}_{i}(\bm{I}-c\bm{P})^{-1}\bm{p}}{\bm{e}^{T}(\bm{I}-c\bm{P})^{-1}\bm{p}}, 
\end{equation}
where $\bm{e}_{i}$ is $n$-dimensional vector with the $i_{th}$ element is 1 and 0 others. 
 
\subsection{Chebyshev Polynomial}
Chebyshev Polynomial is a set of polynomials defined as
\begin{center}
	\small
	$T_{n}(x)=\cos{(n\arccos{x})}$, $n=0,1,2,\cdots$.
\end{center}
$T_{n}(x)$ satisfies: 
\begin{enumerate}
	\small
	\item[$(1)$] $\vert T_{n}(x) \vert \le 1$.
	\item[$(2)$] $\frac{2}{\pi}\int_{-1}^{1}\frac{T_{m}(x)T_{n}(x)}{\sqrt{1-x^{2}}}dx=\left\{ 
	\begin{array}{ll} 
		2,           & if \quad m=n=0,\\
		\delta_{mn}, & else.
	\end{array} 
	\right.$
	\item[$(3)$] $T_{n+1}(x)=2xT_{n}(x)-T_{n-1}(x)$.
\end{enumerate}

Chebyshev Polynomial is usually utilized in constructing the optimal uniform approximating polynomial. For any given function $f(x) \in C(a,b)$ and $\int_{a}^{b}f^{2}(x)dx<\infty$, it follows that    
\begin{center}
	\small
	$f(x)=\frac{c_{0}}{2}+\sum \limits_{k=1}^{\infty}c_{k}T_{k}(x)$,
\end{center}
where $c_{k}=\frac{2}{b-a}\int_{a}^{b} \frac{1}{\sqrt{1-(\frac{2}{b-a}x-\frac{a+b}{b-a})^2}}f(x)T_{k}(\frac{2}{b-a}x-\frac{a+b}{b-a})dx$, $\frac{c_{0}}{2}+\sum \limits_{k=1}^{n}c_{k}T_{k}(x)$ is the optimal uniform approximating polynomial of $f(x)$. For 
\begin{center}
	\small
	$f(x)=(1-cx)^{-1},x \in (-1,1)$,  
\end{center}
where $c \in (0,1)$ is constant, we have $c_{k}=\frac{2}{\pi}\int_{0}^{\pi} \frac{\cos{(kt)}}{1-c\cos{t}}dt$, $\{c_{k}: k=0,1,2,\cdots\}$ decreases monotonously to 0, and 
\begin{center}
	\small
	$\lim \limits_{M \to \infty} \sum \limits_{k=M}^{\infty} c_{k}=\lim \limits_{M \to \infty} \frac{2}{\pi}\int_{0}^{\pi} \frac{\sum \limits_{k=M}^{\infty} \cos{(kt)}}{1-c\cos{t}}dt=0$.
\end{center} Thus, $\forall \epsilon >0, \exists M_{1} \in \mathbb{N}^{+}, s.t.$, $\forall M > M_{1}$, 
\begin{equation}\label{eq23}
	\small
	\max \limits_{x \in (-1,1)} \vert f(x)-(\frac{c_{0}}{2}+\sum \limits_{k=1}^{M}c_{k}T_{k}(x))\vert < \epsilon.  
\end{equation} 

\section{Computing PageRank by Chebyshev Polynomial approximation}
In this section, Formula \eqref{eq23} is extended from the field of real function to the filed of matrix function. We mainly prove that, $\forall \epsilon>0$, $\exists M_{1} \in \mathbb{N}^{+}$, $s.t.$, $\forall M > M_{1}$, 
\begin{equation}\label{eq31}
	\small
	\vert\vert (\bm{I}-c\bm{P})^{-1}\bm{p} - (\frac{c_{0}}{2}\bm{p}+\sum \limits_{k=1}^{M}c_{k}T_{k}(\bm{P})\bm{p})\vert\vert_{2}<\epsilon.
\end{equation}
To this end, we firstly give two lemmas as below. 
\begin{lemma}\label{lem1}\cite{de2020density}
	Diagonalizable matrices are Zariski dense.
\end{lemma}
Lemma \ref{lem1} implies that, for any given matrix $\bm{P} \in \mathbb{R}^{n \times n}$, there exists diagonalizable matrix sequence $\left \{\bm{P}_{t} : t=0,1,2,\cdots \right\}$, where $\bm{P}_{t} \in \mathbb{R}^{n \times n}$ and $\vert\vert \bm{P}_{t} \vert\vert_{2}=1$, $s.t.$, $\bm{P}_{t} \rightrightarrows \bm{P}$. Thus, we have 
\begin{center}
	\small
	$\lim \limits_{t \to \infty} \vert\vert \bm{P}_{t}-\bm{P} \vert\vert_{2} =0$.
\end{center}
Since both $\sum \limits_{k=0}^{\infty}(c\bm{P}_{t})^{k}$ and $\sum \limits_{k=1}^{M}c_{k}T_{k}(\bm{P}_{t})$ are uniformly convergent with respect to $t$, $\forall \epsilon>0$, $\exists N_{1} \in \mathbb{N}^{+}$, $s.t.$, $\forall t>N_{1}$, 
\begin{center}
	\small
	$\vert\vert \sum \limits_{k=0}^{\infty}(c\bm{P})^{k}\bm{p} - \sum \limits_{k=0}^{\infty}(c\bm{P}_{t})^{k}\bm{p} \vert\vert_{2} < \epsilon$, 
\end{center}
and 
\begin{center}
\small
$\vert\vert \sum \limits_{k=1}^{M}c_{k}T_{k}(\bm{P})\bm{p} - \sum \limits_{k=1}^{M}c_{k}T_{k}(\bm{P}_{t})\bm{p} \vert\vert_{2} < \epsilon$.
\end{center}
Since $(\bm{I}-c\bm{P})^{-1}\bm{p} = \sum \limits_{k=0}^{\infty}(c\bm{P})^{k}\bm{p}$, Formula \eqref{eq31} can be hold by proving 
\begin{equation}\label{eq32}
	\small
	\lim \limits_{t \to \infty} \vert\vert \sum \limits_{k=0}^{\infty}(c\bm{P}_{t})^{k}\bm{p} - (\frac{c_{0}}{2}\bm{p} + \sum \limits_{k=1}^{M}c_{k}T_{k}(\bm{P}_{t})\bm{p}) \vert\vert_{2} < \epsilon.
\end{equation}

\begin{lemma}\label{lem2}
	The whole eigenvalues of $\bm{P}$ are real.
\end{lemma}
\begin{proof}[Proof]
	We first prove that $\bm{D}-c\bm{A}$ is positive definite matrix. $\forall \bm{x} \in \mathbb{R}^{n}$ and $\bm{x}\neq \bm{0}$, it follows that 
	\begin{center}
		\small
		$
		\begin{aligned}
			\bm{x}^{T}(\bm{D}-c\bm{A})\bm{x}=&\sum \limits_{i=1}^{n} {d_{i}x_{i}^{2}}-c\sum \limits_{i=1}^{n} \sum \limits_{j=1}^{n}{a_{ij}x_{i}x_{j}} \\
			=&\frac{1}{2}\sum \limits_{i=1}^{n} \sum \limits_{j=1}^{n}{a_{ij}(x_{i}^{2} + x_{j}^{2}-2cx_{i}x_{j})}. 
		\end{aligned}
		$
	\end{center}
	Since $x_{i}^{2} + x_{j}^{2}-2cx_{i}x_{j}>0, i,j=1,2,\cdots,n$, $\bm{D}-c\bm{A}$ is positive definite. 
	
	Let $\bm{B}=\bm{D}-c\bm{A}$, then 
	\begin{center}
		\small
		$(\bm{B}\bm{P}^{T})^{T}=(\bm{A}-c\bm{A}\bm{D}^{-1}\bm{A})^{T}=\bm{A}-c\bm{A}\bm{D}^{-1}\bm{A}=\bm{B}\bm{P}^{T}$. 
	\end{center}
	 $\forall \bm{y} \in \mathbb{C}^{n}$ and $\bm{y} \neq \bm{0}$, assuming $\bm{y}=\bm{y}_{1}+\mathrm{i}\bm{y}_{2}$, where $\bm{y}_{1}, \bm{y}_{2} \in \mathbb{R}^{n}$, we have 
	\begin{center}
		\small
		$
		\begin{aligned}
			\bm{y}^{*}\bm{B}\bm{y} &=(\bm{y}_{1}-\mathrm{i}\bm{y}_{2})^{T}\bm{B}(\bm{y}_{1}+\mathrm{i}\bm{y}_{2}) \\
			&=\bm{y}_{1}^{T}\bm{B}\bm{y}_{1} + \bm{y}_{2}^{T}\bm{B}\bm{y}_{2} + \mathrm{i}(\bm{y}_{1}^{T}\bm{B}\bm{y}_{2}-\bm{y}_{2}^{T}\bm{B}\bm{y}_{1}) \\
			&=\bm{y}_{1}^{T}\bm{B}\bm{y}_{1} + \bm{y}_{2}^{T}\bm{B}\bm{y}_{2} >0.
		\end{aligned}
		$
	\end{center}
	Denote by $\lambda$ the eigenvalue of $\bm{P}^{T}$, by $\bm{\chi}$ the corresponding eigenvector, then
	\begin{center}
		\small
		$\lambda \bm{\chi}^{*} \bm{B} \bm{\chi} =\bm{\chi}^{*}\bm{B}\bm{P}^{T}\bm{\chi} =\bm{\chi}^{*}(\bm{B}\bm{P}^{T})^{T}\bm{\chi} =(\bm{P}^{T}\overline{\bm{\chi}})^{T}\bm{B}\bm{\chi} =\overline{\lambda}\bm{\chi}^{*} \bm{B} \bm{\chi}.$
	\end{center}
	Since $\bm{\chi}^{*} \bm{B} \bm{\chi} > 0$, we have $\overline{\lambda}=\lambda$, i.e., $\lambda$ is real.
\end{proof}

Since $\bm{P}_{t}$ is diagonalizable, there exist $n$ linearly independent eigenvectors. Denote by $\bm{\chi}_{l}(t)$ the eigenvector of $\bm{P}_{t}$, by $\lambda_{l}(t)=\lambda_{l}^{(1)}(t)+\mathrm{i}\lambda_{l}^{(2)}(t)$ the corresponding eigenvalue, where $\vert\vert \bm{\chi}_{l}(t) \vert\vert_{2}=1$, $\lambda_{l}^{(1)}(t), \lambda_{l}^{(2)}(t) \in \mathbb{R}$ and $l=1,2,\cdots,n$, we have $\vert \lambda_{l}(t) \vert \leq 1$ and 
\begin{center}
	\small
	$\vert\vert \bm{P}\bm{\chi}_{l}(t)-\lambda_{l}(t)\bm{\chi}_{l}(t)\vert\vert_{2}=\vert\vert \bm{P}\bm{\chi}_{l}(t)-\bm{P}_{t}\bm{\chi}_{l}(t)\vert\vert_{2} \leq \vert\vert \bm{P}-\bm{P}_{t}\vert\vert_{2}$. 
\end{center}

Given $l$, since $\mathbb{C}$ is complete and $\{\lambda_{l}(t)\}$ is bounded, there exists convergent subsequence $\{\lambda_{l}(t_{k})\} \subseteq \{\lambda_{l}(t)\}$, let $\lambda_{l} = \lim \limits_{k \to \infty}\lambda_{l}(t_{k})$. Similarly, since $\mathbb{R}^{n}$ is complete and $\{\bm{\chi}_{l}(t_{k})\}$ is bounded, there exists convergent subsequence $\{\bm{\chi}_{l}(t_{k_{s}})\} \subseteq \{\bm{\chi}_{l}(t_{k})\}$, let $\bm{\chi}_{l} = \lim \limits_{s \to \infty}\bm{\chi}_{l}(t_{k_{s}})$. For convenience, we still mark $t_{k_{s}}$ as $t$ in the following. Since
\begin{center}
	\small
	$\lim \limits_{t \to \infty} \vert\vert \bm{P}\bm{\chi}_{l}(t)-\lambda_{l}(t)\bm{\chi}_{l}(t)\vert\vert_{2} \leq \lim \limits_{t \to \infty} \vert\vert \bm{P}-\bm{P}_{t} \vert\vert_{2}=0$,   
\end{center}
it follows that 
\begin{center}
	\small
	$\bm{P}\bm{\chi}_{l} = \lim \limits_{t \to \infty} \bm{P}\bm{\chi}_{l}(t) = \lim \limits_{t \to \infty} \lambda_{l}(t)\bm{\chi}_{l}(t) = \lambda_{l} \bm{\chi}_{l}$,   
\end{center}
i.e, $\lambda_{l}$ is the eigenvalue of $\bm{P}$ and $\bm{\chi}_{l}$ is the corresponding eigenvector. According to Lemma \ref{lem2}, $\lambda_{l}$ is real, thus for arbitrary $l=0,1,\cdots,n$, 
\begin{equation} \label{eq34}
	\lim \limits_{t \to \infty} \lambda_{l}^{(2)}(t) = 0. 
\end{equation}

For any given $t$, $\{\bm{\chi}_{l}(t): l=1, 2, \cdots, n\}$ is basis of $\mathbb{R}^{n}$, by which $\bm{p}$ can be linearly represented. Assuming
\begin{center}
	\small
	$\bm{p}=\sum \limits_{l=1}^{n} \alpha_{l}(t)\bm{\chi}_{l}(t)$,
\end{center}
where $\alpha_{l}(t) \in \mathbb{R}$ and $\vert \alpha_{l}(t) \vert < 1$, then
\begin{center}
  \small
  $\sum \limits_{k=0}^{\infty}(c\bm{P}_{t})^{k} \bm{p} = 
   \sum \limits_{k=0}^{\infty}(c\bm{P}_{t})^{k} \sum \limits_{l=1}^{n} \alpha_{l}(t) \bm{\chi}_{l}(t) =
   \sum \limits_{l=1}^{n} \alpha_{l}(t) \sum \limits_{k=0}^{\infty}(c\lambda_{l}(t))^{k} \bm{\chi}_{l}(t)$. 
\end{center}
Since both $\sum \limits_{k=0}^{\infty}(c\lambda_{l}(t))^{k}$ and $\sum \limits_{k=0}^{\infty}(c\lambda_{l}^{(1)}(t))^{k}$ are uniformly convergent with respect to $t$, we have 
\begin{center}
\small
$
\begin{aligned}
   &\vert\vert 
	\sum \limits_{l=1}^{n} \alpha_{l}(t) \sum \limits_{k=0}^{\infty}(c\lambda_{l}(t))^{k} \bm{\chi}_{l}(t) -
	\sum \limits_{l=1}^{n} \alpha_{l}(t) \sum \limits_{k=0}^{\infty}(c\lambda^{(1)}_{l}(t))^{k} \bm{\chi}_{l}(t)
	\vert\vert_{2} \\
  =&\vert 
	\sum \limits_{l=1}^{n} \alpha_{l}(t) (\sum \limits_{k=0}^{\infty}(c\lambda_{l}(t))^{k} -
	                                      \sum \limits_{k=0}^{\infty}(c\lambda^{(1)}_{l}(t))^{k})
	\vert 
	\cdot
	\vert\vert \bm{\chi}_{l}(t) \vert\vert_{2} \\
  =&\vert 
    \sum \limits_{l=1}^{n} \alpha_{l}(t)
    (\sum \limits_{k=0}^{\infty}(c(\lambda^{(1)}_{l}(t)+\mathrm{i}\lambda^{(2)}_{l}(t)))^{k}-\sum \limits_{k=0}^{\infty}(c\lambda^{(1)}_{l}(t))^{k})
    \vert, 
\end{aligned}
$
\end{center}
from Formula \eqref{eq34}, $\forall \epsilon > 0$, $\exists N_{2} \in \mathbb{N}^{+}$, $s.t$, $\forall t > N_{2}$, 
\begin{center}
  \small
  $\vert\vert \sum \limits_{l=1}^{n} \alpha_{l}(t) \sum \limits_{k=0}^{\infty}(c\lambda_{l}(t))^{k} \bm{\chi}_{l}(t) - \sum \limits_{l=1}^{n} \alpha_{l}(t) \sum \limits_{k=0}^{\infty}(c\lambda^{(1)}_{l}(t))^{k} \bm{\chi}_{l}(t) \vert\vert_{2} < \epsilon$, 
\end{center}
i.e., 
\begin{center}
	\small
	$\vert\vert \sum \limits_{k=0}^{\infty}(c\bm{P}_{t})^{k} \bm{p} - \sum \limits_{l=1}^{n} \alpha_{l}(t) \sum \limits_{k=0}^{\infty}(c\lambda^{(1)}_{l}(t))^{k} \bm{\chi}_{l}(t) \vert\vert_{2} < \epsilon$. 
\end{center}
Since $(1-c\lambda^{(1)}_{l}(t))^{-1}=\sum \limits_{k=0}^{\infty}(c\lambda^{(1)}_{l}(t))^{k}$, it follows that 
\begin{equation}\label{eq35}
\small
\vert\vert \sum \limits_{k=0}^{\infty}(c\bm{P}_{t})^{k} \bm{p} - \sum \limits_{l=1}^{n} \alpha_{l}(t) (1-c\lambda^{(1)}_{l}(t))^{-1} \bm{\chi}_{l}(t) \vert\vert_{2} < \epsilon. 
\end{equation}

Approximating $f(\lambda^{(1)}_{l}(t))=(1-c\lambda^{(1)}_{l}(t))^{-1}$ by Chebyshev Polynomial, from Formula (\ref{eq23}), $\forall \epsilon > 0$, $\exists M_{1} \in \mathbb{N}^{+}$, $s.t.$, $\forall M > M_{1}$, 
\begin{center}
	\small
	$
	\vert\vert \sum \limits_{l=1}^{n} \alpha_{l}(t)(1-c\lambda_{l}^{(1)}(t))^{-1}\bm{\chi}_{l}(t) - \sum \limits_{l=1}^{n}\alpha_{l}(t)(\frac{c_{0}}{2}+\sum \limits_{k=1}^{M}c_{k}T_{k}(\lambda_{l}^{(1)}(t)))\bm{\chi}_{l}(t) \vert\vert_{2} < \epsilon
	$, 
\end{center}
i.e.,
\begin{center}
	\small
	$
	\vert\vert \sum \limits_{l=1}^{n} \alpha_{l}(t)(1-c\lambda_{l}^{(1)}(t))^{-1}\bm{\chi}_{l}(t) - (\frac{c_{0}}{2}\bm{p} + \sum \limits_{k=1}^{M}c_{k}\sum \limits_{l=1}^{n}T_{k}(\lambda_{l}^{(1)}(t))\alpha_{l}(t)\bm{\chi}_{l}(t)) \vert\vert_{2} < \epsilon
	$. 
\end{center}
$T_{k}(x)$ is a polynomial of degree $k$, assuming 
\begin{center}
\small
$T_{k}(x)=\sum \limits_{i=0}^{k}a_{i}x^{i}$, 
\end{center}
where the coefficient $a_{i} \in \mathbb{R}$, we have 
\begin{center}
\small
$
\begin{aligned}
 T_{k}(\bm{P}_{t})\bm{p} &= \sum \limits_{i=0}^{k}a_{i}\bm{P}^{i}_{t}\bm{p}   \\
&= \sum \limits_{i=0}^{k}a_{i}\bm{P}^{i}_{t} \sum \limits_{l=1}^{n}\alpha_{l}(t)\bm{\chi}_{l}(t)  \\
&= \sum \limits_{l=1}^{n}\alpha_{l}(t) \sum \limits_{i=0}^{k}a_{i}\bm{P}^{i}_{t}\bm{\chi}_{l}(t)  \\
&= \sum \limits_{l=1}^{n}\alpha_{l}(t) \sum \limits_{i=0}^{k}a_{i}(\lambda^{(1)}_{l}(t)+\mathrm{i}\lambda^{(2)}_{l}(t))^{i}\bm{\chi}_{l}(t).  
\end{aligned}
$
\end{center}
From Formula \eqref{eq34}, $\forall k$ and $\forall \epsilon > 0$, $\exists N_{3} \in \mathbb{N}^{+}$, $s.t.$, $\forall t > N_{3}$,  
\begin{center}
\small
$\vert\vert \sum \limits_{l=1}^{n}\alpha_{l}(t) \sum \limits_{i=0}^{k}a_{i}(\lambda^{(1)}_{l}(t)+\mathrm{i}\lambda^{(2)}_{l}(t))^{i}\bm{\chi}_{l}(t) -
\sum \limits_{l=1}^{n}\alpha_{l}(t) \sum \limits_{i=0}^{k}a_{i}(\lambda_{l}^{(1)}(t))^{i}\bm{\chi}_{l}(t) \vert\vert_{2} < \epsilon$, 
\end{center}
i.e., 
\begin{center}
\small
$\vert\vert
 T_{k}(\bm{P}_{t})\bm{p} -\sum \limits_{l=1}^{n}T_{k}(\lambda_{l}^{(1)}(t))\alpha_{l}(t)\bm{\chi}_{l}(t) 
 \vert\vert_{2} < \epsilon$.  
\end{center}
Thus, $\forall M$ and $\forall \epsilon > 0$, $\exists N_{4} \in \mathbb{N}^{+}$, $s.t.$, $\forall t > N_{4}$, 
\begin{center}
\small
$\vert\vert
\sum \limits_{l=1}^{n} \alpha_{l}(t)(1-c\lambda_{l}^{(1)}(t))^{-1}\bm{\chi}_{l}(t)-
(\frac{c_{0}}{2}\bm{p} + \sum \limits_{k=1}^{M}c_{k}T_{k}(\bm{P}_{t})\bm{p})
\vert\vert_{2} < \epsilon$.  
\end{center}

From Formula \eqref{eq35}, $\forall \epsilon>0$, $\exists M_{1} \in \mathbb{N}^{+}$, $s.t.$, $\forall M>M_{1}$, $\exists N \in \mathbb{N}^{+}$, $s.t.$, $\forall t>N$, 
\begin{center}
\small
$\vert\vert \sum \limits_{k=0}^{\infty}(c\bm{P}_{t})^{k} \bm{p} - (\frac{c_{0}}{2}\bm{p}+\sum \limits_{k=1}^{M}c_{k}T_{k}(\bm{P}_{t})\bm{p}) \vert\vert_{2} < \epsilon$, 
\end{center}
i.e., 
\begin{center}
	\small
	$\lim \limits_{t \to \infty} \vert\vert \sum \limits_{k=0}^{\infty}(c\bm{P}_{t})^{k}\bm{p} - (\frac{c_{0}}{2}\bm{p} + \sum \limits_{k=1}^{M}c_{k}T_{k}(\bm{P}_{t})\bm{p}) \vert\vert_{2}=0$.  
\end{center}
Thus Formula \eqref{eq32} is hold, so as Formula \eqref{eq31}. 

It should be noted that, Formula \eqref{eq31} requires the whole eigenvalues of $\bm{P}$ are real, thus it might not be hold on some directed graphs. From the perspective of polynomial approximation, FP approximates $f(x)=(1-cx)^{-1}$ by taking $\{1,x,x^{2},\cdots\}$ as the basis, while Formula \eqref{eq31} approximates $f(x)$ by taking Chebyshev Polynomial $\{T_{k}(x):k=1,2,\cdots\}$ as basis. Since Chebyshev Polynomial are orthogonal, Formula \eqref{eq31} is more efficient.

\section{Algorithm and analysis}
In this section, a parallel PageRank algorithm for undirected graph is proposed and then the theoretical analysis is given. 

\subsection{Algorithm designing}
From Formula \eqref{eq31}, PageRank vector can be obtained by calculating $\{c_{k}T_{k}(\bm{P})\bm{p}:k=0,1,2.\cdots\}$: 
\begin{enumerate}
	\item [(1)] $\{c_{k} : k=0,1,2,\cdots\}$ relies on damping factor $c$ only, thus we can calculate and store them beforehand.
	\item [(2)] $\{T_{k}(\bm{P})\bm{p} : k=0,1,2,\cdots\}$ can be calculated iteratively by 
	\begin{center}
		\small
		$T_{k+1}(\bm{P})\bm{p}= 2\bm{P}T_{k}(\bm{P})\bm{p}-T_{k-1}(\bm{P})\bm{p}$,    
	\end{center}
    it is sufficed for each vertex reserving the $(k-1)_{th}$ and $k_{th}$ iteration results to generate the $(k+1)_{th}$ iteration result. 
	\item [(3)] During the same iteration round, each vertex calculates independently, thus the computing can be done in parallel.
\end{enumerate}
Restricted by the computing resource, assigning exclusive thread for each vertex is in feasible, we can invoke $K$ independent threads and assigning vertices to them. The Chebyshev Polynomial Approximation Algorithm(CPAA) is given as Algorithm \ref{CPAA}. 

\begin{algorithm}[htbp]
\footnotesize
\caption{ \textbf{Chebyshev Polynomial Approximation Algorithm} }
\begin{algorithmic}[1]
\Require{ \\$K$: The parallelism. 
		  \\$M$: The upper bound of iteration rounds. 
		  \\$\{c_{0}, c_{1}, \cdots, c_{M}\}$: The coefficient of Chebyshev Polynomial approximation.}
\Ensure{$\bm{\pi}$: PageRank vector.}
\State{Each vertex $v_{i}$ maintains $T_{i}, T_{i}^{'}, T_{i}^{''}$ and $\overline{\pi}_{i}$.}
\State{Assign vertices to $K$ calculating threads, denote by $S_{j}$ the set of vertices belonging to thread $j$.}
\State{Initially set $T_{i}=1$ and $\overline{\pi}_{i}=\frac{c_{0}}{2}T_{i}$.}
\State{}
\For {$k=1;k \le M;k++$}
  \ForAll{$j$} \Comment{[Do in parallel.]}
    \If{$k=1$}
      \For{$u \in S_{j}$}
        \State{$T_{u}^{'}=0$;}
        \For{$v_{i} \in N(u)$}\Comment{[$N(u)$ is the set of $u$'s adjacency vertices.]}
          \State{$T_{u}^{'}=T_{u}^{'}+\frac{T_{i}}{deg(v_{i})}$;}
        \EndFor
        \State{$\overline{\pi}_{u}=\overline{\pi}_{u} + c_{1}*T_{u}^{'}$;}
      \EndFor
    \Else
      \For{$u \in S_{j}$}
        \State{$T^{''}_{u}=0$;}
        \For{$v_{i} \in N(u)$} 
          \State{$T_{u}^{''}=T_{u}^{''} + \frac{T_{i}^{'}}{deg(v_{i})}$;}
        \EndFor
        \State{$T_{u}^{''}=2T_{u}^{''}-T_{u}$;}
        \State{$\overline{\pi}_{i}=\overline{\pi}_{i}+c_{k}*T_{u}^{''}$;}
      \EndFor
    \EndIf
  \EndFor
  \ForAll{$j$} \Comment{[Do in parallel.]}
    \For{$u \in S_{j}$}
      \State{$T_{u}=T_{u}^{'}$;}
      \State{$T_{u}^{'}=T_{u}^{''}$;}
    \EndFor
  \EndFor
\EndFor
\State{Calculate PageRank vector $\bm{\pi}$ by $\pi_{i}=\frac{\overline{\pi}_{i}}{\sum\limits_{i=1}^{n}\overline{\pi}_{i}}$.}
\end{algorithmic}\label{CPAA}
\end{algorithm}

From the perspective of vertex, CPAA is a loop of two stages, i.e., mass generating and mass accumulating. Initially, each vertex holds 1 mass and 0 accumulating mass. During the $k_{th}$ iteration, at generating stage, each vertex generates the $(k+1)_{th}$ iteration mass by reading the $k_{th}$ iteration mass from its adjacency vertices and subtracting its own $(k-1)_{th}$ iteration mass; at accumulating stage, each vertex adds the $(k+1)_{th}$ iteration mass weighted by $c_{k+1}$ to the accumulating mass. In the end, the proportion of accumulating mass of one vertex is its PageRank value. Although the mass each vertex holding may change at each iteration, the total mass of the graph is constant at $n$.

From the whole, the total accumulating mass $S$ of the graph is constant and falls into two parts, i.e., the accumulated and unaccumulated. CPAA is a process that converting the unaccumulated mass to the accumulated mass. Initially, the accumulated mass is 0 and the unaccumulated mass is $S$. While CPAA running, the accumulated mass increases and the unaccumulated mass decreases, specifically, at the $k_{th}$ iteration, the accumulated mass increases by $c_{k}n$ and the unaccumulated mass decreases by the same amount. When the accumulated mass is $S$ and the unaccumulated mass is 0, CPAA gets converged. The distribution of accumulated mass is the PageRank vector. 

It should be noted that, the dependency between the $k_{th}$ iteration and $(k+1)_{th}$ iteration can not be eliminated, paralleling only exists in the same iteration round.

\subsection{Algorithm analysis}
We analysis CPAA from three aspects, convergence rate, error and computation amount. 

\subsubsection{Convergence rate}
Convergence rate reflects the rapidity of unaccumulated mass decreasing to 0. Let 
\begin{center}
	\small
	$
	S_{k} = \left\{ 
	\begin{array}{ll}
		n\frac{c_{0}}{2},                                 & if \quad k=0,\\
		n\frac{c_{0}}{2} + n\sum \limits_{i=1}^{k} c_{i}, & else.
	\end{array}
	\right.
	$
\end{center}
$S_{k}$ is accumulated mass at the end of the $k_{th}$ iteration, and $S=\lim \limits_{k \to \infty} S_{k}$. Let 
\begin{center}
	\small
	$\sigma_{k}=\frac{S-S_{k}}{S-S_{k-1}}, k=1,2,3,\cdots$. 
\end{center}
Then $1-\sigma_{k}$ is the ratio of unaccumulated mass decreasing at the $k_{th}$ iteration. 

\begin{proposition}\label{pos1}
	$\sigma_{k}=\frac{c^{2}-(2-c)(1-\sqrt{1-c^{2}})}{c^{2}-c(1-\sqrt{1-c^{2}})}$, $c \in (0,1)$, $k=1,2,3,\cdots$.
\end{proposition}
\begin{proof}
	
	Firstly, let $B_{k}=\sum \limits_{i=k}^{\infty}c_{k}$, then $B_{k}$ is absolutely convergent, and
	\begin{center}
		\small
		$1-\sigma_{k}=\frac{c_{k}}{B_{k}}$.
	\end{center}
	since $c_{k}=\frac{2}{\pi} \int_{0}^{\pi} \frac{\cos{kt}}{1-c\cos{t}} dt$, it follows the recurrence relation
	\begin{center}
		\small
		$
		c_{k-1}+c_{k+1} = \frac{2}{\pi} \int_{0}^{\pi} \frac{2\cos{kt}\cos{t}}{1-c\cos{t}} dt = \frac{2}{c} c_{k}
		$.
	\end{center}
	 According to
	\begin{center}
		\small
		$
		\begin{aligned}
			& c_{k-1}+c_{k+1} = \frac{2}{c} c_{k}, \\
			& c_{k}  +c_{k+2} = \frac{2}{c} c_{k+1}, \\
			& \cdots,            
		\end{aligned}
		$
	\end{center}
	we have
	\begin{center}
		\small
		$B_{k}+c_{k-1}+B_{k}-c_{k}=\frac{2}{c}B_{k}$,  
	\end{center}
    thus $B_{k}=\frac{c_{k-1}-c_{k}}{\frac{2}{c}-2}$ and $\frac{c_{k}}{B_{k}} = \frac{\frac{2}{c}-2}{\frac{c_{k-1}}{c_{k}}-1}$.
	
	Since $c_{0}=\frac{2}{\sqrt{1-c^{2}}}$, 
	      $c_{1}=\frac{2}{c}(\frac{1-\sqrt{1-c^{2}}}{\sqrt{1-c^{2}}})$, $c_{2}=\frac{2}{c^{2}}\frac{2(1-\sqrt{1-c^{2}})-c^{2}}{\sqrt{1-c^{2}}}$, 
	it follows that 
	\begin{center}
		\small
		$
		\frac{c_{0}}{c_{1}} = \frac{c_{1}}{c_{2}}=\frac{c}{1-\sqrt{1-c^{2}}}   
		$. 
	\end{center}
	According to $\frac{c_{0}}{c_{1}}+\frac{c_{2}}{c_{1}}=\frac{2}{c}$, we have 
	\begin{center}
		\small
		$
		\frac{c_{0}}{c_{1}} = \frac{c_{1}}{c_{2}}= \cdots =\frac{c_{k-1}}{c_{k}}=\frac{c}{1-\sqrt{1-c^{2}}}   
		$. 
	\end{center}
    Thus, $\sigma_{k}=\frac{c^{2}-(2-c)(1-\sqrt{1-c^{2}})}{c^{2}-c(1-\sqrt{1-c^{2}})}$.
\end{proof}

Proposition(\ref{pos1}) shows that $\sigma_{k}$ relates with $c$ only. For any given $c$, $\sigma_{k}$ is a constant. For convenience, we mark $\sigma_{k}$ as $\sigma_{c}$ in the following. While CPAA running, the unaccumulated mass decreases by $1-\sigma_{c}$ proportion per iteration. Referring to the Power method, $\sigma_{c}$ indicates the convergence rate of CPAA. Since
\begin{center}
	\small
	$\frac{\sigma_{c}}{c}=\frac{c-(\frac{2}{c}-1)(1-\sqrt{1-c^{2}})}{c^{2}-c(1-\sqrt{1-c^{2}})}<1$, $c \in (0,1)$, 
\end{center}
CPAA has higher convergence rate compared with the Power method. Figure \ref{cpaa-sigma} shows the relation between $\sigma_{c}$ and $c$. When $c=0.85$, $\sigma_{c}=0.5567$, which implies that CPAA only takes 65\% iteration rounds of Power method to get converged.
\begin{figure}[htbp]
	\centering
	\includegraphics[width=0.48\textwidth,height=0.28\textheight]{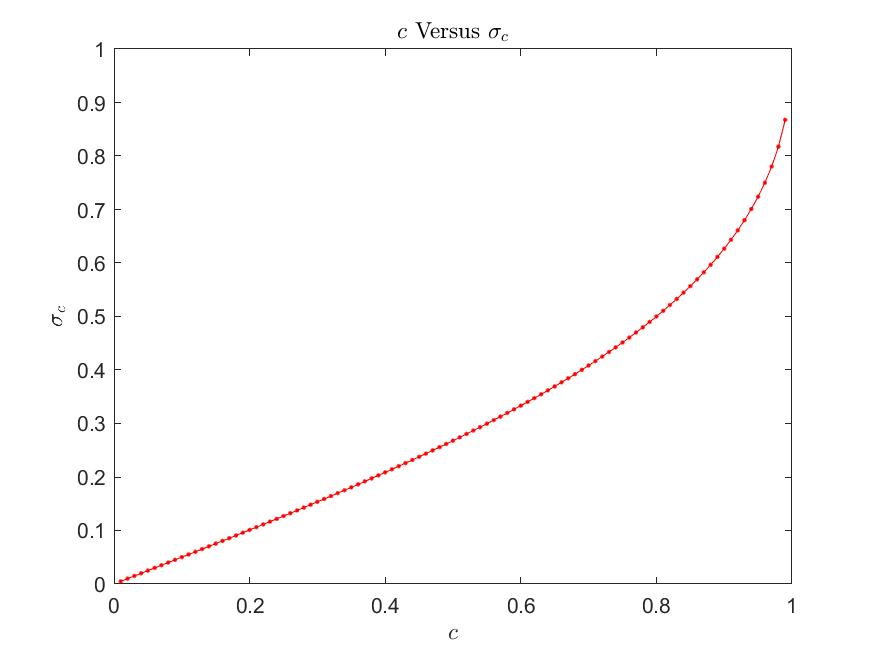}
	\caption{Convergence rate}
	\label{cpaa-sigma}
\end{figure}

\subsubsection{Error}
It is hard to estimate the error of each vertex, we discuss the relation between the relative error and the upper bound of iteration rounds from the whole. Denote by 
\begin{center}
	\small
	$ERR_{M}=1-\frac{S_{M}}{S}$
\end{center}
the relative error, then we have 
\begin{center}
	\small
	$S_{M}=\frac{c_{0}}{2}+\sum \limits_{k=1}^{M}c_{k}=\frac{c_{0}}{2}+\frac{1-\beta^{M}}{1-\beta} \beta c_{0}$, 
\end{center}
thus,
\begin{equation}\label{eq41}
	\small
	ERR_{M}=1-\frac{S_{M}}{S}=\frac{2\beta^{M+1}}{1+\beta}, 
\end{equation}
where $\beta=\frac{1-\sqrt{1-c^{2}}}{c}$. Figure \ref{cpaa-err} shows that $ERR_{M}$ decreases exponentially with respect to $M$, when $c=0.85$, $ERR_{M}$ can be less than $10^{-4}$ within 20 iteration rounds. It should be noted that the error estimation above mentioned is very rough as taking no graph structure into consideration. 
\begin{figure}[h]
	\centering
	\includegraphics[width=0.48\textwidth,height=0.28\textheight]{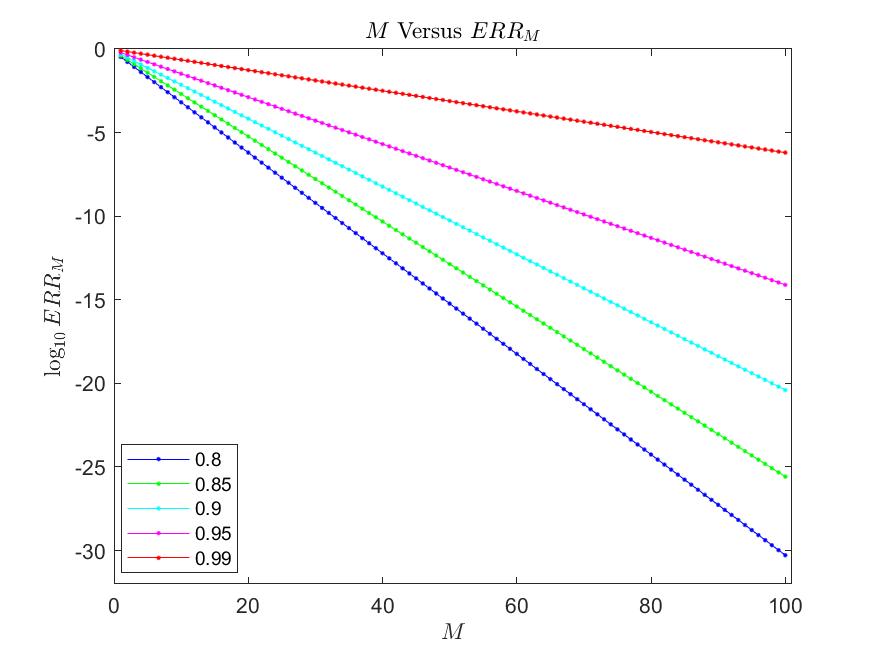}
	\caption{Relative error}
	\label{cpaa-err}
\end{figure}

\subsubsection{Computation amount}
While CPAA running, for vertex $v$, there are $deg(v)$ multiplications and $deg(v)$ additions at generating stage, and $1$ addition at accumulating stage. Therefore, the computation consists of $m$ multiplications and $m+2n$ additions per iteration. The total computation are  $Mm$ multiplications and $M(m+2n)$ additions. 

\section{Experiments}
In this section, we experimentally demonstrate CPAA. Firstly, the hardware and measurements are introduced briefly. Then, the convergence of CPAA is illustrated. At last, the comparison of CPAA with other algorithms is elaborated. 

\subsection{Experimental setting}
All algorithms in this experiment are implemented with C++ on serves with Intel(R) Xeon(R) Silver 4210R CPU 2.40GHz 40 processors and 190GB memory. The operation system is Ubuntu 18.04.5 LTS. Six data sets are illustrated in table \ref{table1}, where $n$, $m$ and $deg=\frac{m}{n}$ represent the number of vertices, the number of edges and the average degree respectively. The CPU time consumption $T$, iteration rounds $k$ and max relative error $ERR=\max \limits_{v_{i} \in V} \frac{\vert \overline{\pi}_{i}- \pi_{i} \vert}{\pi_{i}}$ are taken to estimate the algorithms, where the true PageRank value $\pi_{i}$ is obtained by Power method at the $210_{th}$ iteration. The damping factor $c=0.85$. 

\begin{table}[h]
\footnotesize
\centering
\begin{tabular}{l c c c }
	\toprule
	Data sets           &$n$            &$m$            &$deg$         \\
	\midrule
	NACA0015           &1039183         &6229636        &5.99          \\
	delaunay-n21       &2097152         &12582816       &6.0           \\ 
    M6                 &3501776         &21003872       &6.0           \\
	NLR                &4163763         &24975952       &6.0           \\
	CHANNEL            &4802000         &85362744       &17.78         \\
	kmer-V2            &55042369        &117217600      &2.13          \\
	\bottomrule
\end{tabular}
\caption{Data sets}
\label{table1}
\end{table}

\subsection{Convergence}
The relation of $ERR$, $T$ and $k$ on six data sets is illustrated in Figure \ref{converge}. 
\begin{enumerate}
	\item[(1)] The blue lines show that $ERR$ has a negative exponential relation with $k$, which is in consistent with Formula (\ref{eq41}). With the increasing of iteration rounds, the result gets closer to the true value. The max relative error can be less than $10^{-4}$ within 20 iterations.
	\item[(2)] The red lines show that $T$ has a positive linear relation with $k$, thus, It also implies that the computation of each iteration is constant. On kmer-V2, the max relative error is lower than $10^{-2}$ within 15 seconds. 
	\item[(3)] The blue lines almost stay horizontal after 50 iteration rounds, which seems CPAA has limitation in precision. We believe it is not a algorithm flaw, but caused by the insufficiency of C++'s DOUBLE data type, whose significant digit number is 15. The computer considers $1$ is equal to $1+10^{-16}$, so that, $c_{k}$ with $k>50$ can not change the result.
\end{enumerate}
\begin{figure}[h]
	\centering
	\includegraphics[width=0.32\textwidth,height=0.20\textheight]{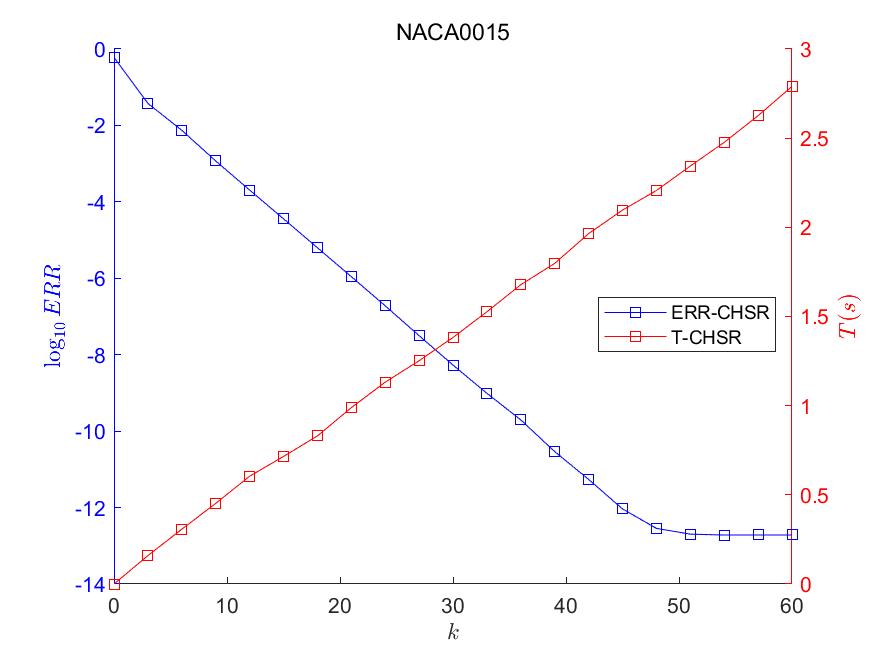}
	\includegraphics[width=0.32\textwidth,height=0.20\textheight]{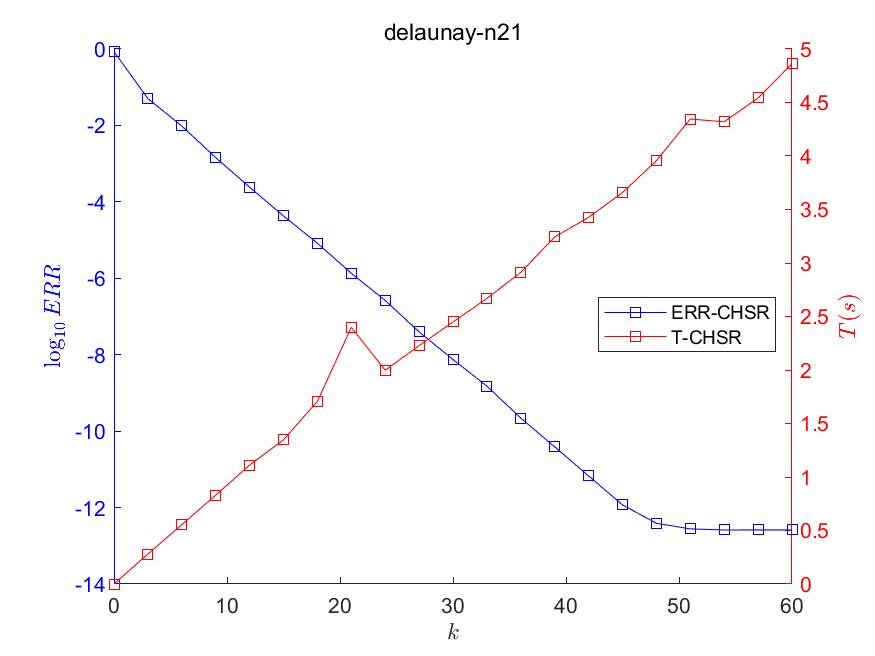}
	\includegraphics[width=0.32\textwidth,height=0.20\textheight]{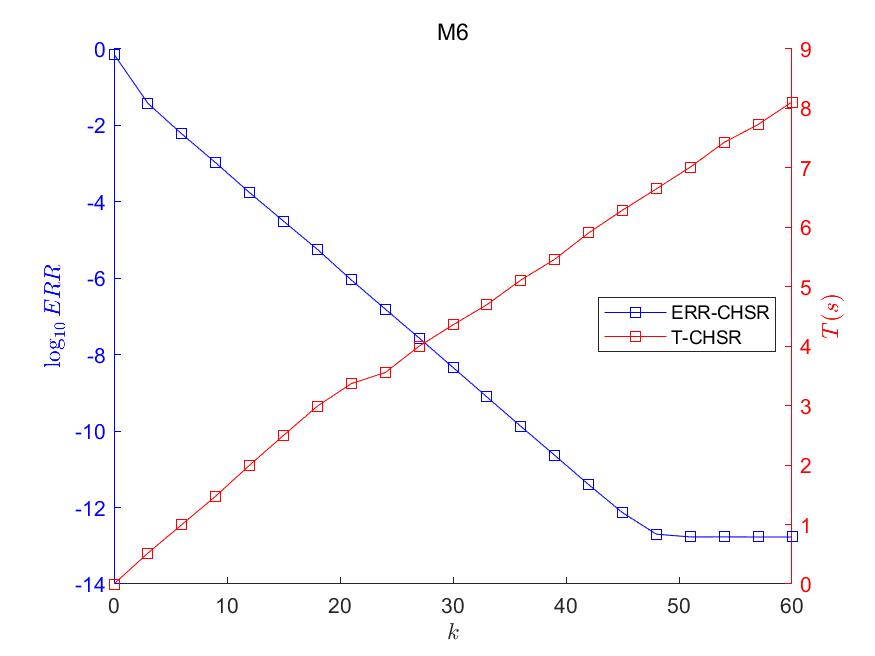}
	\includegraphics[width=0.32\textwidth,height=0.20\textheight]{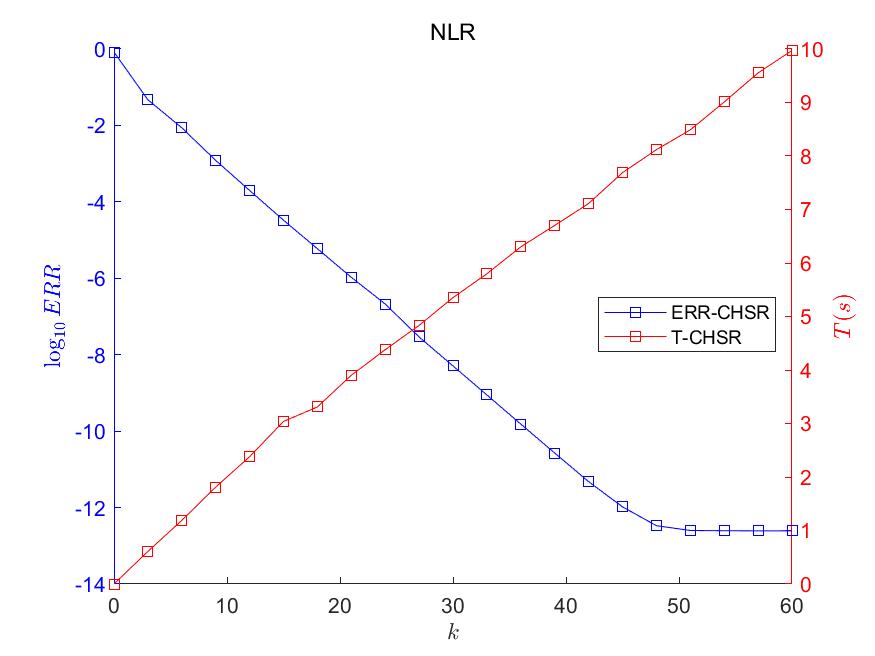}
	\includegraphics[width=0.32\textwidth,height=0.20\textheight]{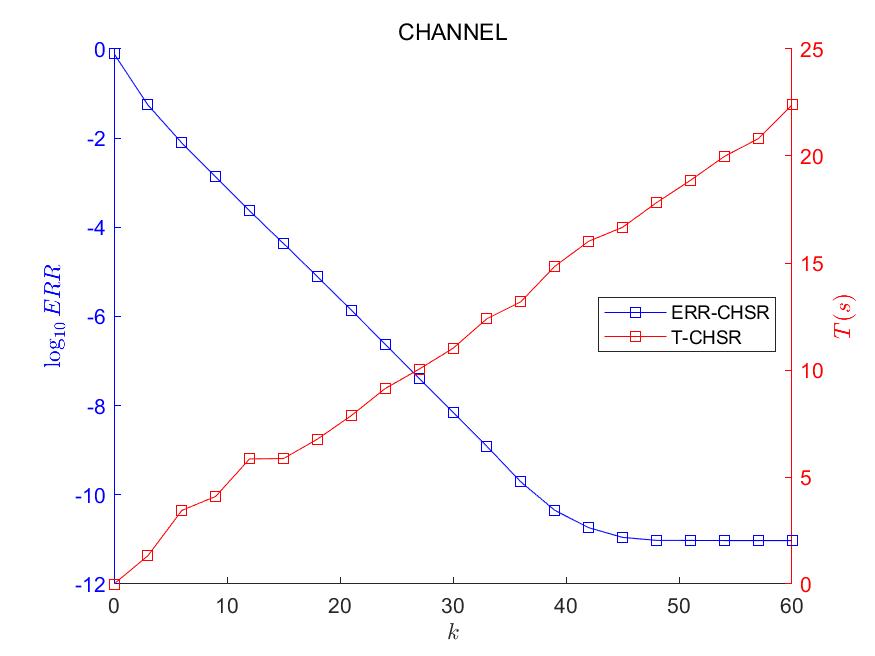}
	\includegraphics[width=0.32\textwidth,height=0.20\textheight]{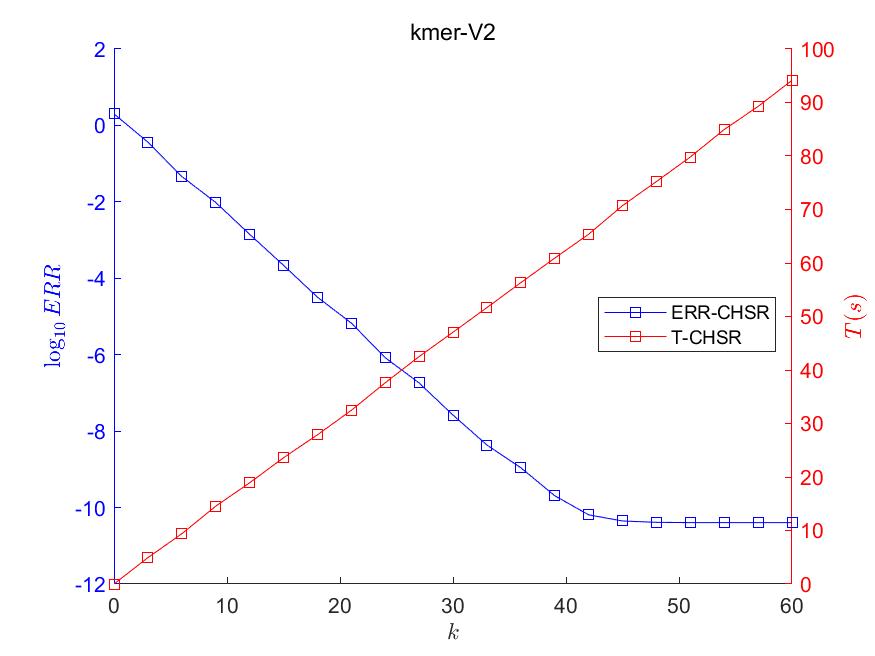}
	\caption{$k$ versus $ERR$ and $T$}
	\label{converge}
\end{figure}

\subsection{Comparison with other algorithms}
We compare CPAA with Power method (SPI) \cite{brin2012reprint,page1999pagerank}, parallel Power method (MPI) \cite{duong2012parallel} and an improvement of FP, i.e., IFP1 \cite{ifp1}. Both MPI, IFP1 and CPAA execute with 38 parallelism. The relation of $ERR$ and $T$ are illustrated in Figures \ref{err}. Table \ref{table3} shows iteration rounds and CPU time consumption when $ERR < 10^{-3}$. 
\begin{enumerate}
	\item[(1)] The green, cyan and magenta lines are lower than the blue lines, which indicates that both MPI, IFP1 and CPAA are much faster than SPI. For example, on kmer-V2, it takes near 900 seconds for SPI to get $ERR<10^{-3}$, while the remaining algorithms cost less than 90 seconds. It implies that parallelization is an effective solution to accelerate PageRank computing. 
	\item[(2)] The magenta lines are lower than the green and cyan lines, which implies that CPAA outperforms MPI and IFP1. For example, on kmer-V2, it takes about 67 seconds for MPI to get $ERR<10^{-3}$, while CPAA costs less than 24 seconds. Since the whole algorithms are executed with the same parallelism, we believe the advantages are owe to the higher convergence rate of CPAA. 
	\item[(3)] On almost all six data sets, Power method takes 20 iterations to get $ERR < 10^{-3}$, while CPAA only takes 12 iterations, and can be at most 39 times faster than SPI. It implies that CPAA takes 60\% iterations of Power method to get converged, which is consistent with the theoretical analysis. 
\end{enumerate}

\begin{figure}[h]
	\centering
	\includegraphics[width=0.32\textwidth,height=0.20\textheight]{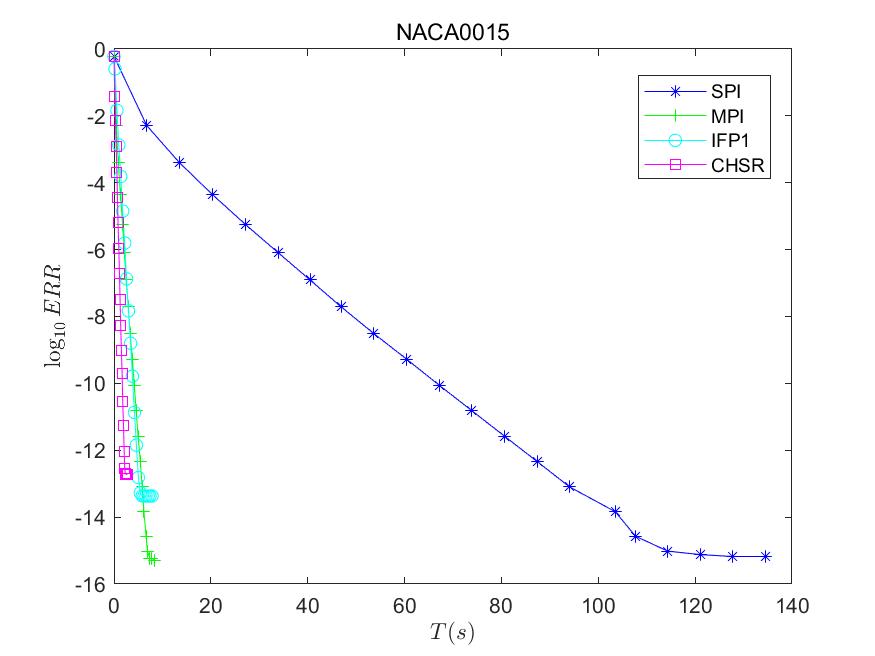}
	\includegraphics[width=0.32\textwidth,height=0.20\textheight]{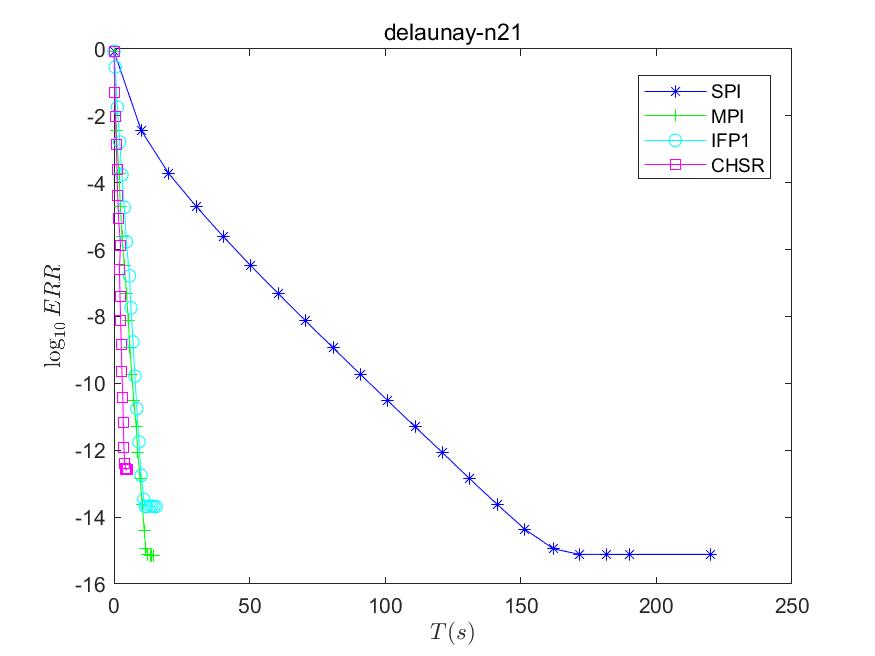}
	\includegraphics[width=0.32\textwidth,height=0.20\textheight]{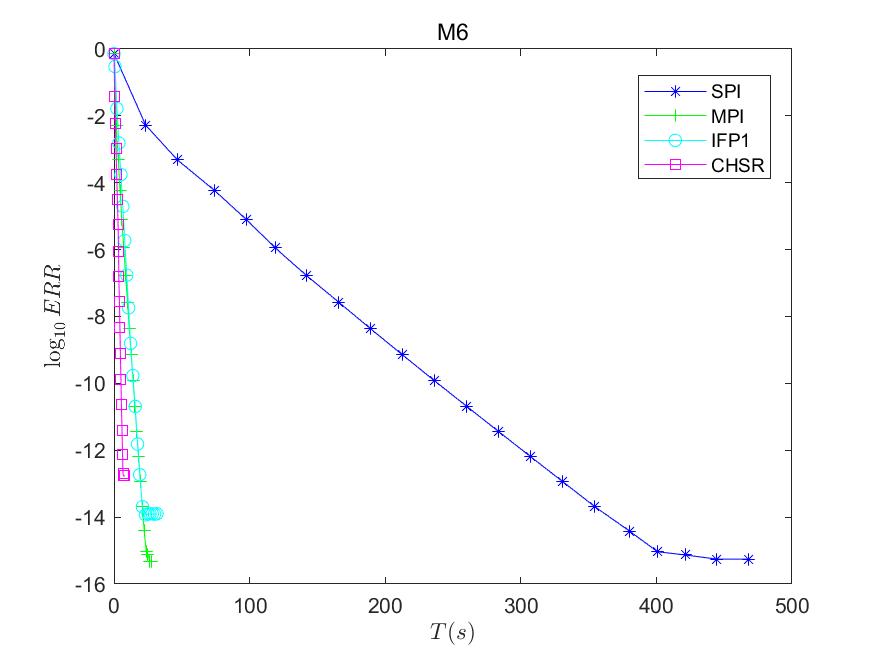}
	\includegraphics[width=0.32\textwidth,height=0.20\textheight]{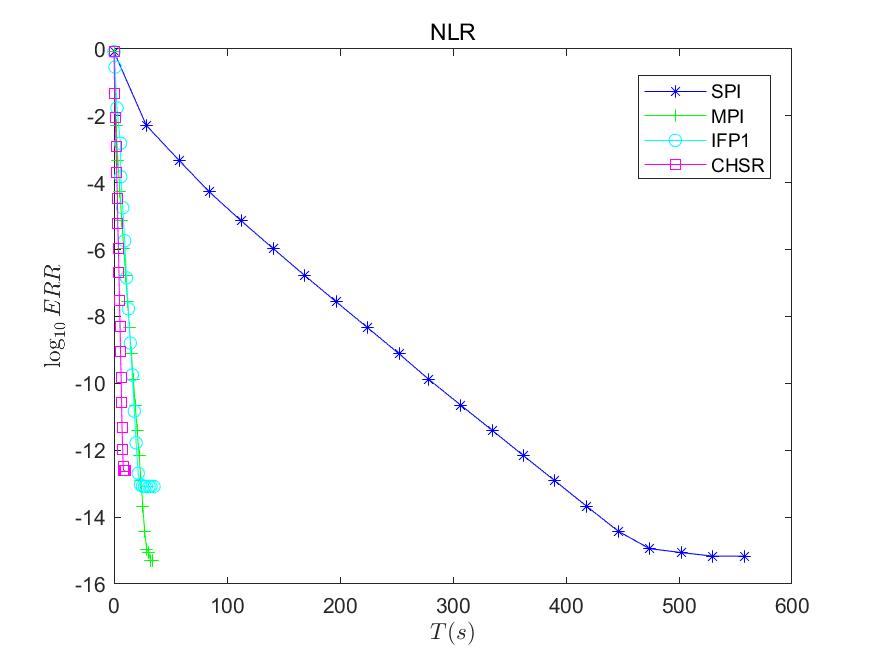}
	\includegraphics[width=0.32\textwidth,height=0.20\textheight]{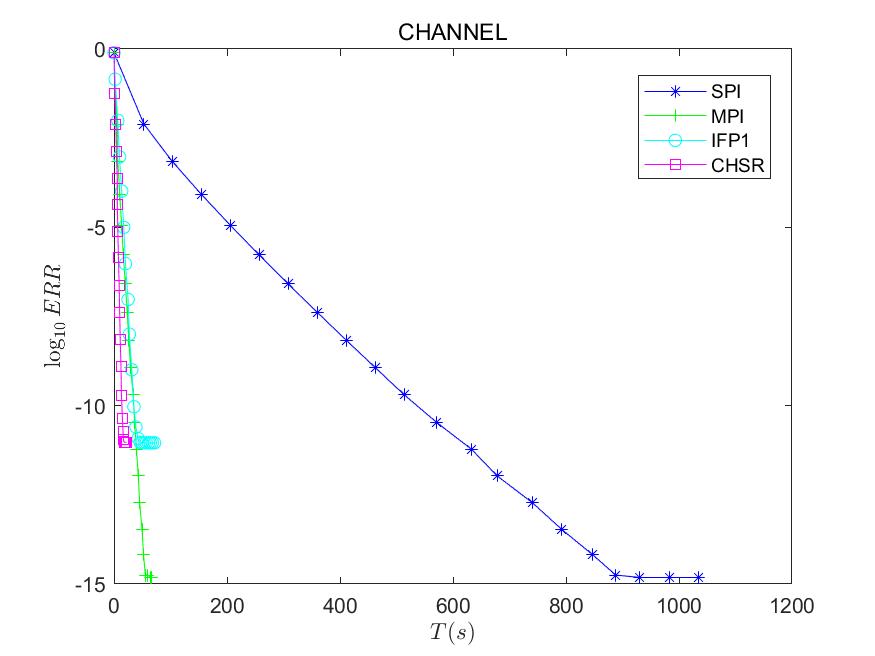}
	\includegraphics[width=0.32\textwidth,height=0.20\textheight]{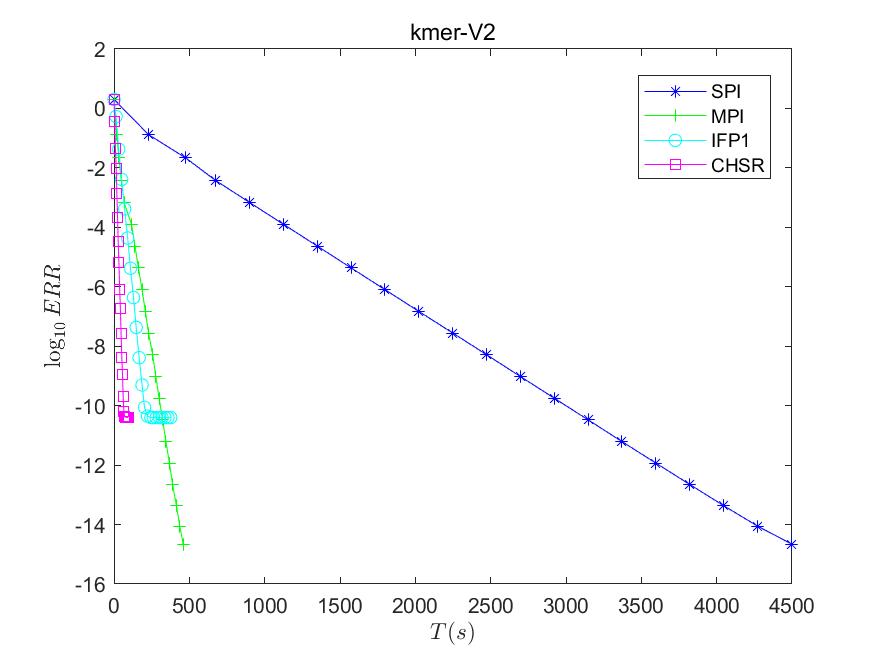}
	\caption{$T$ versus $ERR$}
	\label{err}
\end{figure}

\begin{table}[h]
\footnotesize
\centering
\begin{tabular}{l c c c c c c c c}
	\toprule
	\multirow{2}*{Data sets}   &\multicolumn{2}{c}{SPI}   &\multicolumn{2}{c}{MPI}  &\multicolumn{2}{c}{ITA}   &\multicolumn{2}{c}{CPAA}  \\
	\cline{2-9}
	               &$k$    &$T$         &$k$      &$T$          &$k$      &$T$          &$k$               &$T$                \\
	\midrule
	NACA0015       &20     &13.541      &20       &0.887        &-        &1.390        &\textbf{12}       &\textbf{0.604}       \\
	delaunay-n21   &20     &20.175      &20       &1.854        &-        &2.921        &\textbf{12}       &\textbf{1.111}       \\
	M6             &20     &46.771      &20       &3.130        &-        &5.097        &\textbf{12}       &\textbf{1.999}        \\
	NLR            &20     &57.876      &20       &3.467        &-        &5.983        &\textbf{12}       &\textbf{2.379}        \\
	CHANNEL        &20     &102.828     &20       &6.819        &-        &9.689        &\textbf{12}       &\textbf{5.845}        \\
	kmer-V2        &40     &899.406     &40       &67.068       &-        &70.346       &\textbf{15}       &\textbf{23.597}       \\
	\bottomrule
\end{tabular}
\caption{Iterations rounds and CPU time consumption when $ERR<10^{-3}$}
\label{table3}
\end{table}

\section{Conclusion}
How to compute PageRank more efficiently is an attractive problem. In this paper, a parallel PageRank algorithm specially for undirected graph called CPAA is proposed via Chebyshev Polynomial approximation. Both the theoretical analysis and experimental results indicate that CPAA has higher convergence rate and less computation amount compared with the Power method. Moreover, CPAA implies a more profound connotation, i.e., the symmetry of undirected graph and the density of diagonalizable matrix. Based on the latter, the problem of computing PageRank is converted to the problem of approximating $f(x)=(1-cx)^{-1}$. This is a innovative approach for PageRank computing. In the future, some other orthogonal polynomials\textemdash Laguerre polynomial, for example\textemdash can be taken into consideration.  
\bibliographystyle{elsarticle-num} 
\bibliography{ref}

\begin{thebibliography}{10}
\expandafter\ifx\csname url\endcsname\relax
  \def\url#1{\texttt{#1}}\fi
\expandafter\ifx\csname urlprefix\endcsname\relax\def\urlprefix{URL }\fi
\expandafter\ifx\csname href\endcsname\relax
  \def\href#1#2{#2} \def\path#1{#1}\fi

\bibitem{brin2012reprint}
S.~Brin, L.~Page, Reprint of: The anatomy of a large-scale hypertextual web
  search engine, Computer networks 56~(18) (2012) 3825--3833.

\bibitem{gleich2015pagerank}
D.~F. Gleich, Pagerank beyond the web, Siam Review 57~(3) (2015) 321--363.

\bibitem{brown2017pagerank}
S.~Brown, A pagerank model for player performance assessment in basketball,
  soccer and hockey, arXiv preprint arXiv:1704.00583 (2017).

\bibitem{Yin2020change}
Ying, Hou, Zhou, Dou, Wang, Shao, The changes of central cultural cities based
  on the analysis of the agglomeration of literati's footprints in tang and
  song dynasties, Journal of Geo-information Science 22~(5) (2020) 945--953.

\bibitem{haveliwala2003computing}
T.~Haveliwala, S.~Kamvar, D.~Klein, C.~Manning, G.~Golub, Computing pagerank
  using power extrapolation, Tech. rep., Stanford (2003).

\bibitem{avrachenkov2007monte}
K.~Avrachenkov, N.~Litvak, D.~Nemirovsky, N.~Osipova, Monte carlo methods in
  pagerank computation: When one iteration is sufficient, SIAM Journal on
  Numerical Analysis 45~(2) (2007) 890--904.

\bibitem{kamvar2004adaptive}
S.~Kamvar, T.~Haveliwala, G.~Golub, Adaptive methods for the computation of
  pagerank, Linear Algebra and its Applications 386 (2004) 51--65.

\bibitem{kamvar2003extrapolation}
S.~D. Kamvar, T.~H. Haveliwala, C.~D. Manning, G.~H. Golub, Extrapolation
  methods for accelerating pagerank computations, in: Proceedings of the 12th
  international conference on World Wide Web, 2003, pp. 261--270.

\bibitem{wu2007power}
G.~Wu, Y.~Wei, A power-arnoldi algorithm for computing pagerank, Numerical
  Linear Algebra with Applications 14~(7) (2007) 521--546.

\bibitem{zhu2005distributed}
Y.~Zhu, S.~Ye, X.~Li, Distributed pagerank computation based on iterative
  aggregation-disaggregation methods, in: Proceedings of the 14th ACM
  international conference on Information and knowledge management, 2005, pp.
  578--585.

\bibitem{sarma2013fast}
A.~D. Sarma, A.~R. Molla, G.~Pandurangan, E.~Upfal, Fast distributed pagerank
  computation, in: International Conference on Distributed Computing and
  Networking, Springer, 2013, pp. 11--26.

\bibitem{das2013distributed}
A.~Das~Sarma, D.~Nanongkai, G.~Pandurangan, P.~Tetali, Distributed random
  walks, Journal of the ACM (JACM) 60~(1) (2013) 1--31.

\bibitem{2019Distributed}
S.~Luo, Distributed pagerank computation: An improved theoretical study,
  Proceedings of the AAAI Conference on Artificial Intelligence 33 (2019)
  4496--4503.

\bibitem{luo2020improved}
S.~Luo, Improved communication cost in distributed pagerank computation--a
  theoretical study, in: International Conference on Machine Learning, PMLR,
  2020, pp. 6459--6467.

\bibitem{2006Local}
R.~Andersen, F.~R.~K. Chung, K.~J. Lang, Local graph partitioning using
  pagerank vectors, in: 47th Annual IEEE Symposium on Foundations of Computer
  Science (FOCS 2006), 21-24 October 2006, Berkeley, California, USA,
  Proceedings, 2006.

\bibitem{2011Chebyshev}
D.~I. Shuman, P.~Vandergheynst, P.~Frossard, Chebyshev polynomial approximation
  for distributed signal processing, in: Distributed Computing in Sensor
  Systems, 7th IEEE International Conference and Workshops, DCOSS 2011,
  Barcelona, Spain, 27-29 June, 2011, Proceedings, 2011.

\bibitem{Berkhin2005A}
Berkhin, Pavel, A survey on pagerank computing, Internet Mathematics 2~(1)
  (2005) 73--120.

\bibitem{de2020density}
R.~J. de~la Cruz, P.~Saltenberger, Density of diagonalizable matrices in sets
  of structured matrices defined from indefinite scalar products, arXiv
  preprint arXiv:2007.00269 (2020).

\bibitem{page1999pagerank}
L.~Page, S.~Brin, R.~Motwani, T.~Winograd, The pagerank citation ranking:
  Bringing order to the web., Tech. rep., Stanford InfoLab (1999).

\bibitem{duong2012parallel}
N.~T. Duong, Q.~A.~P. Nguyen, A.~T. Nguyen, H.-D. Nguyen, Parallel pagerank
  computation using gpus, in: Proceedings of the Third Symposium on Information
  and Communication Technology, 2012, pp. 223--230.

\bibitem{ifp1}
Q.~Zhang, R.~Tang, Z.~Yao, J.~Liang, Two parallel pagerank algorithms via
  improving forward push, arXiv preprint arXiv:2302.03245 (2023).

\end{thebibliography}

\end{document}